\documentclass[preprint,5p,times,onecolumn]{elsarticle}
\biboptions{numbers,square,comma,sort&compress}

\usepackage{amsmath}
\usepackage{amsthm}
\usepackage{amssymb}
\usepackage{fancyhdr}
\usepackage[normalem]{ulem}
\usepackage[usenames,dvipsnames]{color}

\theoremstyle{plain}
\newtheorem{theorem}{Theorem}[section]

\theoremstyle{definition}
\newtheorem{definition}{Definition}[section]

\usepackage{graphicx}
\usepackage{graphics}

\newcommand{\comment}[1]{}

\newcommand{\marisa}[1]{\textcolor{red}{$\langle${\slshape{\bfseries Marisa:} #1}$\rangle$}}

\journal{Neurocomputing}
\begin{document}
\begin{frontmatter}

%\begin{document}
\title{Parameter identifiability and identifiable combinations in \\generalized Hodgkin-Huxley models}
%\author{Olivia J. Walch$^*$ and Marisa C. Eisenberg$^{*\dagger}$
%\blfootnote{Email addresses: ojwalch@umich.edu, marisae@umich.edu. $^*$Department of Mathematics, University of Michigan, Ann Arbor. $^\dagger$Department of Epidemiology, School of Public Health, University of Michigan, Ann Arbor.}
%}
\author[UMmath]{Olivia J. Walch}
\ead{ojwalch@umich.edu}
\author[UMmath,UMepi]{Marisa C. Eisenberg}
\ead{marisae@umich.edu}
\address[UMmath]{Department of Mathematics, University of Michigan, Ann Arbor}
\address[UMepi]{Department of Epidemiology, School of Public Health, University of Michigan, Ann Arbor}
%\maketitle

\begin{abstract} The use of Hodgkin-Huxley (HH) equations abounds in the literature, but the identifiability of the HH model parameters has not been broadly considered. Identifiability analysis addresses the question of whether it is possible to estimate the model parameters for a given choice of measurement data and experimental inputs. Here we explore the structural identifiability properties of a generalized form of HH from voltage clamp data. Through a scaling argument, we conclude that the steady-state gating variables are not identifiable from voltage clamp data, and then further show that their product together with the conductance term forms an identifiable combination. We additionally show that these parameters become identifiable when the initial conditions for each of the gating variables are known. The time constants for each gating variable are shown to be identifiable, and a novel method for estimating them is presented. Finally, the exponents of the gating variables are shown to be identifiable in the two-gate case, and we conjecture these to be identifiable in the general case. These results are broadly applicable to models using HH-like formalisms, and show in general which parameters and combinations of parameters are possible to estimate from voltage clamp data. 
\end{abstract}


\end{frontmatter}

%Misc things to cite
%Biological mechanism and identifiability of a class of stationary conductance model for Voltage-gated Ion channels
%Febe Francis, M�riam R. Garc�a, Oliver Mason, Richard H. Middleton
%
%Markov models for ion channels: versatility versus identifiability and speed
%Martin Fink , Denis Noble
%
%stuff from: https://scholar.google.com/scholar?client=safari&rls=en&oe=UTF-8&um=1&ie=UTF-8&lr&cites=311005647246782393

\section{Introduction}

Since its introduction in 1952, the Hodgkin and Huxley (HH) model for membrane excitability in the squid giant axon has become one of the most commonly used formalisms in mathematical neuroscience, with citations now numbering in the tens of thousands \cite{HH1952, HH1952b}. By partitioning membrane voltage change into currents caused by the flow of distinct ions, Hodgkin and Huxley created an illuminating characterization of the underlying cause of axon potentials. In the model, each ionic current is gated by channels, and the probability of the channels being open or closed is voltage-dependent. The original model assumes these gates operate independently, and, while subsequent work has shown this not to be the case, HH nonetheless provides a good description of ionic behavior at the appropriate scale and remains highly relevant in the literature today. Consequently, much work has been dedicated to parameter estimation for the HH equations \cite{Buhry2011,Lee2006,Csercsik2010,Willms1999,Fink2009,Willms2002,Hafner1981,Vavoulis2012}.

Most treatments of HH parameter estimation have tackled the problem with a focus on practicality---estimating parameters given noisy and limited data. However, there has been relatively little examination \cite{Csercsik2010} of the more basic but essential question of structural identifiability: given perfect, noise-free data, can the parameters in the model be uniquely determined? While such perfect data is of course unrealistic, structural identifiability is a prerequisite for practical identifiability and successful parameter estimation. Furthermore, such structural identifiability information can be used to generate insights into ways to reduce the model to improve identifiability, or to guide collection of new data that will allow all parameters to be estimated. Thus, understanding the structural identifiability properties of the HH model provides an important foundation in efforts to connect HH-based models with data. 

Here, we examine the identifiability of a broad class of generalized HH-type models. We elucidate the identifiable combination structure for this class of models, evaluate the role of initial conditions in identifiability, and consider what additional data is needed to ensure identifiability. Additionally, we show that the proof of identifiability of the time-constants for the gating variables allows us to develop a novel practical estimation approach for general HH-type models.

\section{Methods}

\subsection{Generalized Hodgkin-Huxley Models}
The HH equations for ionic current can be generalized for $p_n$ ion channel gates of type $n$ acting independently as
\begin{equation}
I(t) = g (V-E) m_1^{p_1}  \cdots m_n^{p_n} 
\label{eq:current}
\end{equation}

\noindent where $g$ is the conductance associated with the ion channel, $V$ is the voltage of the cell, $E$ is the reversal potential of the ion, and the $m_i$ terms represent the probability of a voltage-controlled gate being open. Each of the $m_i$ is further taken to satisfy the differential equation
\begin{equation}
\frac{dm_i}{dt} = \frac{m_{i,\infty}(V)-m_i}{\tau_i(V)}
\label{eq:gates}
\end{equation}

\noindent in which $m_{i,\infty}(V)$ is the steady-state probability of the gate being open when the voltage is held at $V$ and $\tau_i(V)$ is the time-constant for the kinetics of the gate's activation or inactivation at that same voltage. In cases similar to the classical HH model, where only two types of gate appear, the conventional notation $m^{p_1}h^{p_2}$ may be used instead of $m_1^{p_1}m_2^{p_2}$.

While the HH model represents a heavily approximated version of ionic channel dynamics (assuming all ion channels are independent, ignoring changes in reversal potential due to ion flow), its ability to reproduce action potentials and other properties of cell electrophysiology have led to it remaining highly relevant over the six decades since its publication.
 
Typically, the voltage-dependent parameters, $m_{i,\infty}$ and $\tau_i$, are estimated from voltage clamp experiments. In a voltage clamp, a feedback loop is used to hold voltage at a constant value, and the current required to maintain this constant voltage (theoretically, exactly cancelling the ionic currents) is recorded. Individual currents are isolated, either by blocking all other ionic currents, or by subtracting traces where the current in question is blocked from those where nothing is blocked. Once found, the values for $m_{i,\infty}$ and $\tau_i$ across all the fixed voltages are considered together and fit so that the two parameters are then described by functions of voltage, $m_{i,\infty}(V)$ and $\tau_i(V)$. These functions often follow standard forms, e.g. Boltzmann equations, although these are not necessarily completely physically accurate \cite{Willms1999}. 

Much other work has concerned the process of parameter estimation for HH-type models \cite{Buhry2011}, but only a few sources have addressed the issue of identifiability. In \cite{Csercsik2010}, the identifiability of the parameters $\{g,  m_{1,\infty}, m_{2,\infty}\}$ was evaluated in currents of the form $I(t) = g \, m_1^{p_1} m_2^{p_2}(V-E)$. Csercsik and colleagues show that these parameters are unidentifiable, and moreover, no pair of them is identifiable, although the precise form of any identifiable combinations is not determined. Here, we repeat and extend that analysis in the generalized case for an arbitrary number of gates, using a scaling argument, and then additionally show that the time constants $\tau_i$ are identifiable. We also examine the identifiability of powers $p_i$ in the `two independent gates'-type scenario, and evaluate how knowledge of the initial conditions of the gating variables alters the identifiability structure of the model.

\subsection{Identifiability and differential algebra}

Identifiability addresses the question of whether the a given set of parameters can be uniquely estimated for a given model and data. Structural identifiability addresses this question in the case where we assume `perfect,' noise-free data (i.e. complete knowledge of the measured variables for all time points). While this represents an unrealistic best-case scenario, it forms a necessary condition for estimation from real, noisy data, and indeed structural unidentifiability is quite common among mechanistic models \cite{Audoly2001, Meshkat2009, Eisenberg2013}. The importance of identifiability and its place as a necessary precursor to fitting data are discussed further in \cite{Csercsik2010, Eisenberg2013, miao2011}.

Methods for determining structural identifiability have been developed in detail elsewhere \cite{Audoly2001, Meshkat2009, Eisenberg2013, Bellu2007, Bellman1970, Chapman2009, Cobelli1980}, so we provide only brief overview here. Consider a model of the form:

\[ \textbf{x}'(t,\textbf{p})=f(\textbf{x}(t,\textbf{p}),\textbf{u}(t),t; \textbf{p}),\]
\[\textbf{y}(t,\textbf{p})=g(\textbf{x}(t,\textbf{p});\textbf{p}),\]

\noindent where $\textbf{p}$ represents the (vector of) parameters, $\textbf{x}$ is the unobserved state variable vector, $\textbf{u}(t)$ are the known experimental input(s) into the system, if any, and $\textbf{y}(t)$ represents the observed (measured) output (s). We also let $\bf x_0$ represent the vector of initial conditions for $\bf x (t)$. A model is said to be \emph{identifiable} if $\textbf{p}$ can be recovered uniquely from $\textbf{y}$ and  $\textbf{u}$. Because there may be particular or degenerate parameter values and initial conditions for which an otherwise identifiable model is unidentifiable (e.g. initial conditions starting at a constant steady state), structural identifiability is often defined for almost all parameter values and initial conditions \cite{Meshkat2009, Audoly2001, Saccomani2003}:

\begin{definition} For a given ODE model $\dot{\textbf{x}} = f(\textbf{x},t,\textbf{u,p})$ and output  $\textbf{y}$, an individual parameter $p$ is \emph{uniquely (globally) structurally identifiable} if for almost every value $\textbf{p}^*$ and almost all initial conditions, the equation $\textbf{y}(\textbf{x},t,\textbf{p}^*) = \textbf{y}(\textbf{x},t,\textbf{p})$ implies $p = p^*$.  A parameter $p$ is said to be \emph{non-uniquely (locally) structurally identifiable} if for almost any $\textbf{p}^*$ and almost all initial conditions, the equation $\textbf{y}(\textbf{x},t,\textbf{p}^*) = \textbf{y}(\textbf{x},t,\textbf{p})$ implies that $p$ has a finite number of solutions.  \end{definition}

\begin{definition}  Similarly, a model $\dot{\textbf{x}} = f(\textbf{x},t,\textbf{u,p})$ is said to be \emph{uniquely} (respectively \emph{non-uniquely}) \emph{structurally identifiable}  for a given choice of output $\textbf{y}$ if every parameter is uniquely (respectively non-uniquely) structurally identifiable, i.e. the equation $\textbf{y}(\textbf{x},t,\textbf{p}^*) = \textbf{y}(\textbf{x},t,\textbf{p})$ has only one solution, $\bf p = p^*$ (respectively finitely many solutions).   \end{definition}

There are a number of approaches to determining identifiability; here, we use the differential algebra approach \cite{Ollivier1990, Audoly2001, Eisenberg2013} which is briefly summarized as follows. For models with $f$ and $g$ rational, construct the $\textit{input-output}$ equations from the state variable equations and the output equations. Input-output 
equations are monic differential polynomials in the input and output variables and their derivatives with rational coefficients in the parameter vector $\textbf{p}$ (i.e. with the state variables $\textbf{x}$ and all of their derivatives eliminated from the equations). These can be generated in many ways, including using Ritt's pseudo division or Groebner bases, among others \cite{Ollivier1990, Audoly2001,meshkat2012alternative,Eisenberg2013, bearup2013input, Cobelli1980}.
The coefficients (rational in $\textbf{p}$) of the input-output equations are identifiable, and the structural identifiability of the model (i.e. injectivity of the map from parameters to output), can then be tested simply by checking injectivity of the map from the parameters to the coefficients.

As a simple example for illustrative purposes, we consider the HH model given in Eqs. \eqref{eq:current} and \eqref{eq:gates} in the minimal case where $n=p_1=1$. Then solving for $m_1$ from Eq. \eqref{eq:current} yields:
$$m_1 = \frac{I(t)}{g\,(V-E)}.$$
Plugging this into Eq. \eqref{eq:gates} yields
$$\frac{\dot{I}(t)}{g\,(V-E)} = \frac{m_{1,\infty} - \frac{I(t)}{g\,(V-E)}}{\tau_1}.$$
To make this equation monic, we simply clear the coefficient for $\dot{I}$, yielding our input-output equation:
$$0 = \dot{I}(t) - \frac{g\,m_{1,\infty}}{\tau_1}(V-E) + \frac{1}{\tau_1} I(t).$$
The coefficients of the input-output equation are identifiable, so that we see that $\tau_i$ is an identifiable parameter, while $g\,m_{1,\infty}$ forms an identifiable combination with neither parameter identifiable individually.

\section{Results and Discussion}
\subsection{Generalized Hodgkin-Huxley equation identifiability}
As stated above, we consider the identifiability of a generalized form of Hodgkin-Huxley equations, given in Eqs. \eqref{eq:current} and \eqref{eq:gates}.
Unless otherwise stated, we assume we are fitting a single voltage clamp trace and therefore that $V$ is fixed and known. Our output is thus given by $y = I(t)$. Voltage steps in clamp experiments typically are preceded by a period of time in which the voltage is held fixed at a holding potential, $V_{hold}$, consistent across all trials; when this value is used, it will always be distinguished from the step value $V$. Typically, the reversal potential $E_{ion}$ is readily determined through experimental means \cite{HH1952,HH1952b} while the other parameters (the $m_i$, $\tau_i$, and $g$) are estimated from the data. 

\subsubsection{$m_{\infty}$ combinations and non-identifiability}
The authors of \cite{Csercsik2010} show that for the two-gate HH model, $I(t) = g \, m^{p_1} h^{p_2}(V-E)$, no pair from $\{ g, m_{\infty}, h_{\infty}\}$ is identifiable. It is possible to show the same is true in the generalized case with a simple scaling argument, much like what appears as an example in \cite{Eisenberg2013}. 

\begin{theorem}
The conductance term $g$ and the steady-state parameters $m_{i,\infty}$ are not identifiable from voltage clamp data. Nor is the product of any strict subset of $g$ and the $m_{i,\infty}$; however, the product $g \prod_{i=1}^{n} m_{i,\infty}^{p_i}$ is an identifiable combination. %\marisa{should we explicitly note in the theorem that no sub-product would be identifiable either? (the analog of their 'no pair is identifiable' thing)}
\end{theorem}
\label{th:unid}

\begin{proof} From Eq. \eqref{eq:current} for ionic current,
\[I(t) = g (V-E) m_1^{p_1}  \cdots m_n^{p_n}\]
observe that we can rescale $m_1$ by $g^{\frac{1}{p_1}}$ to get $gm_1^{p_1} = \hat{m}_1^{p_1} $ and $I(t) = \hat{m}_1^{p_1}m_2^{p_2}  \cdots m_n^{p_n} (V-E)$. This new 
expression for the ionic current has the same identifiability and input-output structure as the previous one except the conductance term does not appear; hence $g$ is not an identifiable parameter (as $g$ can take on any value and yield the same output given the same input, by adjusting the value of $m_{1,\infty}$).

The argument for the steady-state parameters proceeds similarly. Assuming no steady-state parameter is exactly zero, rescale $m_i$ for $i = 1, \, \cdots,\, n-1$ by $\frac{1}{m_{i,\infty}}$ so that $z_i = m_i/m_{i,\infty}$. Then 

\[ I = g (V-E) m_n^{p_n} \prod_{i=1}^{n-1} z_i^{p_i} m_{i,\infty}^{p_i}.\]

\noindent Next, rescale $m_n$ so that $z_n = m_n \prod_{i=1}^{n-1} m_{i,\infty}^{p_i}$, so that

\[ I= g  (V-E) z_1 \cdots z_n,\] and 

\[ \frac{d z_i}{dt} = \frac{1 - z_i}{\tau_i}\] for $i = 1, \, \cdots, \, n-1$ while 

\[\frac{dz_n}{dt} = \frac{\prod_{i=1}^{n} m_{i,\infty}^{p_i} - z_n}{\tau_n}.\]

Again the identifiability structure is unchanged, but the steady-state parameters appear only once, grouped into a single term:  $\prod_{i=1}^{n} m_{i,\infty}^{p_i}$. 
The individual steady-state parameters are thus not identifiable, nor is the product of any strict subset of the steady-state parameters and conductance term. Their full product with $g$ is identifiable because \[ \lim_{t \rightarrow \infty} \frac{I(t)}{V-E} = g\prod_{i=1}^{n} m_{i,\infty}^{p_i}\] which, under the assumption of perfect data, is known. \end{proof}
\bigskip

\begin{figure*}
\begin{center}
\includegraphics[width = 0.9\textwidth]{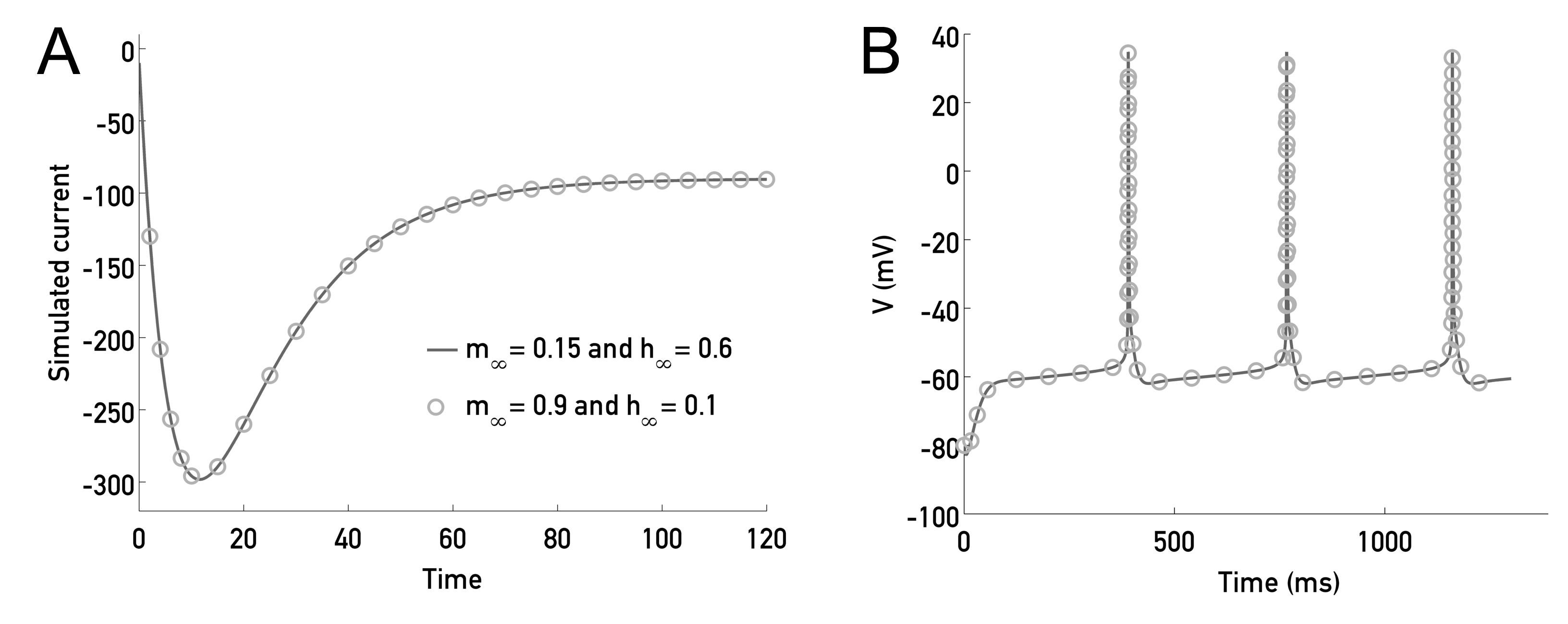}
\caption{A. Two different sets of $(m_0, m_{\infty})$, $(h_0, h_{\infty})$  pairs yield the same current trace in a simulated sodium current $I_{Na} = m^3 h g (V - E)$.  B. With appropriate scaling of $m_0$, $m_{\infty}$, $h_0$ and $h_{\infty}$, the sodium current of the Forger-Sim SCN neuron model and output firing of the model as a whole are identical for two different parameter sets (shown as circles and as a solid line).}
\label{fig:unid}
\end{center}
\end{figure*}

To illustrate this issue, Figure \ref{fig:unid}A shows two simulated sodium-type current traces, with gating variables of the form $m^3h$, which are identical despite different $m_{\infty}$ and $h_{\infty}$ values; a similar example is included in \cite{Csercsik2010}.

\subsubsection{$\tau$ identifiability}
While the undentifiability of the steady-state parameters can be ascertained through scaling, to show the identifiability of the time constants using results from differential algebra requires a slightly more technical analysis. %\mstrike{mathematical trickery.} \marisa{sorry! you're still the best, mathematical trickery}

\begin{theorem}
The time constants for the gating variable kinetics, $\tau_i$, are identifiable from voltage clamp data. 
\end{theorem}
\label{th:tau}

\begin{proof}We start by considering the case where $p_i = 1$ for all $i$. Rescale the current trace by the steady state value $I_{\infty} = g (V-E)  \prod_{i=1}^{n} m_{i,\infty}$. This value, while possibly very small, is non-zero. Next, let $m_i/m_{i,\infty} = z_i$ as previously and denote the rescaled current by $\hat{I}(t) = I(t)/I_{\infty}$. Make the substitution $1 - \tau_i\dot{z_i} = 
z_i$ to rewrite this current as 

\[\hat{I}(t) =  \prod_{k=1}^n (1 - \tau_k\dot{z_k}).\]

\noindent Expanding this product gives: 

\begin{align*}
\hat{I}(t) \, =  \, &1 - \sum_{i}  \tau_i\dot{z_i} + \sum_{1 \leq i<j \leq n} \tau_i \tau_j \dot{z_i} \dot{z_j} \,  + \, \cdots \\ 
&   +  \, \sum_{1 \leq j_1<\cdots<j_k \leq n} (-1)^k \tau_{j_1} \cdots \tau_{j_k} \dot{z_{j_1}} \cdots \dot{z_{j_k}} + \,  \cdots \,  \\
&  + \, (-1)^n \tau_{1} \cdots \tau_{n} \dot{z_1} \cdots \dot{z_n}  
\end{align*}

\medskip
\noindent Since $\ddot{z_i} = -\dot{z_i}/\tau_i$, we can write down successive derivatives of $I(t)$ as:

{\small
\begin{align*}
(-1)^{\ell+1}  \hat{I}^{(\ell)}(t) \, = \, &  - \sum_{i} \left( \frac{1}{\tau_i}\right)^{\ell}  \tau_i\dot{z_i} \, + \sum_{1 \leq i<j \leq n}  \left( \frac{1}{\tau_i} +  \frac{1}
{\tau_j}\right)^{\ell} \tau_i \tau_j \dot{z_i} \dot{z_j} + \, \cdots \, \\
& +  \sum_{1 \leq j_1<\cdots<j_k \leq n} (-1)^k \left( \frac{1}{\tau_{j_1}} + \cdots + \frac{1}{\tau_{j_k}}\right)^{\ell}  \tau_{j_1} \cdots \tau_{j_k} \dot{z_{j_1}} \cdots \dot{z_{j_k}}\\
& + \cdots
 + \, \,   (-1)^n \left( \frac{1}{\tau_{1}} + \cdots + \frac{1}{\tau_{n}}\right)^{\ell}\tau_{1} \cdots \tau_{n} \dot{z_1} \cdots \dot{z_n} 
\end{align*}
}

\noindent From the first $2^n -1$ derivatives of $\hat{I}(t)$, we can write down the system:
{\small
\[\begin{pmatrix}
1 & 1 & \cdots & 1 \\
\lambda_1 & \lambda_2 & \cdots & \lambda_{2^n-1} \\
\lambda_1^2 & \lambda_2^2 & \cdots & \lambda^2_{2^n-1} \\
\vdots & \vdots & \vdots & \vdots \\
\lambda_1^{2^n - 2} & \lambda_2^{2^n -2} & \cdots & \lambda^{2^n-2}_{2^n-1} \\
\end{pmatrix}
\begin{pmatrix}
- \tau_1 \dot{z_1} \\
\vdots \\
(-1)^{k} \tau_{j_1} \cdots \tau_{j_k} \dot{z_{j_1}} \cdots \dot{z_{j_k}}\\
\vdots \\
(-1)^{n} \tau_{1} \cdots \tau_{n} \dot{z_1} \cdots \dot{z_n} 
\end{pmatrix}
= \begin{pmatrix}
\tilde{I} \\
\hat{I}'(t) \\
\vdots \\
\hat{I}^{(p)}(t) \\
\vdots \\
\hat{I}^{(2^n-1)}(t) \\
\end{pmatrix}
\]
}
where $\tilde{I} = \hat{I}(t) -1$ and the $\lambda_{i}$'s represent a sequential indexing of the possible $\sum -\frac{1}{\tau}$ quantities with $\lambda_i = -\frac{1}
{\tau_i}$ for $i = 1 \cdots n$ and \[ \lambda_{2^n - 1} =-\sum_i^n  \frac{1}{\tau_i}. \] Denote this Vandermonde matrix by $V$ and this system as $Vx = y$, with

\[x = \begin{pmatrix}
- \tau_1 \dot{z_1} \\
\vdots \\
(-1)^{k} \tau_{j_1} \cdots \tau_{j_k} \dot{z_{j_1}} \cdots \dot{z_{j_k}}\\
\vdots \\
(-1)^{n} \tau_{1} \cdots \tau_{n} \dot{z_1} \cdots \dot{z_n} 
\end{pmatrix}\]

\noindent The `$\lambda$'-Vandermonde matrix is invertible as long as $\lambda_i \neq \lambda_j$, which occurs with probability 1, so $x = V^{-1}y$ (with $i$th entry denoted $x_i
$). We can next observe that

\[ \prod_{k=1}^n x_k = (-1)^n \tau_1 \cdots \tau_n \dot{z}_1 \cdots \dot{z}_n =  x_{2^n-1}.\]

 Since all $x_i$ can be written as a linear combination of the derivatives of $\hat{I}(t)$, the equation $\prod_{k=1}^n x_k  =  x_{2^n-1}$ gives a polynomial in $\hat{I}(t)$ and its 
derivatives with coefficients in the parameters: an input-output equation. 

Furthermore, the coefficients of the monomials of the form $\hat{I}^{(k)}(t)$ (or $\tilde{I}$, ``singletons'') are the entries of the last row of the inverse Vandermonde matrix. Note that 
the entries of the last row are also the coefficients $a_k$ of the polynomial $\sum_{k=0}^{2^n-2} a_k x^k$, the $2^n-2$ roots of which are $\lambda_1 \cdots \lambda_
{2^n-2}$ ($\lambda_{2^n-1}$ can be recovered by summing over all the roots and dividing by $2^n-1$). These roots are invariant under scalings of the coefficients; hence, by 
finding the roots of the polynomial with coefficients taken from these monomials, we can recover the set of $\tau_i$ from them. This implies that 
the $\tau_i$s are identifiable parameters. 

It remains to consider the case where $p_i$ is not necessarily 1. Let $N = \sum p_i$, and note that we can use a $(2^N - 1)$-by-$(2^N - 1)$ Vandermonde matrix in 
writing a linear system $Vx = y$ similar to the one above, with the key difference being that now certain $\lambda$ are equal. $V$ in this case will not be invertible; 
however, by eliminating the repeated columns and collapsing all duplicate entries of $x$ into single entries (e.g. Replacing $x_1 = x_2 = x_3 = \tau_1 \dot{z_1} $ with 
$x_1 = 3 \tau_1 \dot{z_1} $ ), we can rewrite the system so that $V$ is $(2^N - 1)$-by-$p$, $x$ is $p$-by-$1$, and $y$ is $(2^N - 1)$-by-1, where $p$ is a quantity that emerges from the partitioning of $N$ into $p_i$. Removing rows of $V$ until it is square (while preserving the first $n$ and last rows), we can construct an input-output equation by equating entries of $x$ in the same way as before, and the coefficients of singleton monomials will also be the coefficients of a polynomial with 
zeros equal to $\lambda_1, \cdots, \lambda_{p-1}$. \end{proof}

Hence, the time constants are identifiable, even in the generalized case discussed here. This proof can also provide a way to estimate the time constants from experimental data, discussed further below.

\subsubsection{Power identifiability}
\begin{theorem}
For a classical two-gate Hodgkin-Huxley model of the form $I = g (V  - E) m_1^{p_1} m_2^{p_2}$, the powers $p_1$ and $p_2$ are identifiable.
\end{theorem}
This can be shown by considering successive derivatives of $\log(I(t))$, specifically $f(t) = I'(t)/I(t)$, $f'(t)$, and $f''(t)$. Computing $f$ and its derivatives and 
replacing $m_i''(t)$ with $-m_i'(t)/\tau_i$ yields expressions in terms of the parameters $p_i$ and $\tau_i$ and the ratio of state variables $m_i'(t)/m_i(t)$. 
Replacing these ratios with $a = m_1'(t)/m_1(t)$ and $b = m_2'(t)/m_2(t)$ gives the following:
{\small
\[f(t) = a p_1 + b p_2\]
\[f'(t) =  -a^2 p_1 - (a p_1)/\tau_1 -b^2 p_2 - (b p_2)/\tau_2 \]
\[ f''(t) = 2 a^3 p_1 + (a p_1)/\tau_1^2 + (3 a^2 p_1)/\tau_1 +2 b^3 p_2 + (b p_2)/\tau_2^2 + (3 b^2 p_2)/\tau_2 \]
}
Using Mathematica to eliminate the variables $a$ and $b$ through a Groebner basis computation yields an input-output equation, the coefficients of which readily imply the identifiability of $p_1$ and $p_2$. We conjecture that a similar result holds for the generalized case, however the growth of this Groebner basis calculation's complexity with $n$ has made this somewhat intractable so far.

\subsection{Consideration of initial conditions}
Thus far we have assumed no knowledge of the initial conditions of the model (although initial conditions for the output $I(t)$ and its derivatives are assumed known as $I(t)$ is measured perfectly for all times). In this case, the additional information provided by knowledge of the initial conditions changes the identifiability structure of the problem. 

\begin{theorem}
If initial conditions for the gating variables $m_{i,0} := m_i(0)$ are known, the steady state parameters at a fixed voltage, $m_{i,\infty}(V)$, are identifiable from voltage clamp data.
\end{theorem}
\begin{proof} As before, scale the original current trace $I(t)$ by $g (V-E) \prod_{k=1}^{n} m_{k,\infty}^{p_k} $ to get $\hat{I}(t) = \prod_{k=1}^{n} \hat
{m}_k^{p_k}$. The only parameters in this scaled model are the (identifiable) $\tau_i$ values, so by Theorem 2.2 the model is itself identifiable. We can solve for the 
initial conditions of the scaled model using the explicit solution $\hat{m}_i(t) = 1 - (1 - \hat{m}_{i,0})\exp(-t/\tau_i)$. Once found, the scaled initial conditions 
can be divided into the unscaled initial conditions, yielding the steady state parameters $m_{i,\infty}(V)$. \end{proof}

\subsubsection{Identifiable combinations in terms of initial conditions}
Given the lack of identifiability for HH models unless the gating variable initial conditions are known, a natural question arises in whether consideration of the initial conditions---even when unknown---might yield additional identifiable combinations. Moreover, when practically fitting the model, the initial conditions of the gating variables would need to be considered.

\begin{theorem} The pairs $m_{i,0}/m_{i,\infty}$ are identifiable combinations given voltage clamp data. 
\end{theorem}
\label{th:ICcombos}

\begin{proof}
%Given
 %\[I(t) = g (V-E) \prod_{i = 1}^n m_i^{p_i}\]
%\noindent r
Replacing the $m_i$ in Eq. \eqref{eq:current} with their explicit solutions and factoring yields
%\begin{align*}
%I = & g (V-E) \left(\prod_j m_{j,\infty}^{p_j} \right) \cdot \\
%& \prod_{i = 1}^n (1 - (1 - \frac{m_{i,0}}{m_{i,\infty}})\exp(-t/\tau_i)^{p_i})
%\end{align*}
\[
I(t) =  g (V-E) \left(\prod_j m_{j,\infty}^{p_j} \right) \cdot  \prod_{i = 1}^n \left(1 - \left(1 - \frac{m_{i,0}}{m_{i,\infty}}\right)\exp(-t/\tau_i)\right)^{p_i}
\]
Dividing the current trace by its steady-state value therefore gives us

\[ \hat{I}(t) = \prod_{i = 1}^n \left(1 - \left(1 - \frac{m_{i,0}}{m_{i,\infty}}\right)\exp(-t/\tau_i)\right)^{p_i}\]

\noindent in which the ratios $m_{i,0}/m_{i,\infty}$ are identifiable, along the with $\tau_i$'s. 
\end{proof}

We note that this proof also shows that the initial conditions for the gating variables are identifiable for the scaled model considered in the proof of Theorem \ref{th:unid}, wherein we scaled all $m_i$ by their steady state values (as in this case, the steady state values are precisely the $m_{i,0}/m_{i,\infty}$).

In applying these results to experimental data, we also note that it is reasonable to assume that the initial conditions of the gating variables $m_{i,0}$ are the same for all experimental voltage steps $V$ because of the pre-step fixed holding potential $V_{hold}$. As a result,  the shape of the $m_{i,\infty}(V)$ curve up to scaling by the constant $m_{i,0}$ can be found from the data in this way, and a concrete value of $m_{i,0}$ can be chosen so that $\frac{m_{i,0}}{m_{i,\infty}(V_{hold})} = 1$.  
If the curve isn't smooth at a certain value of $V$, it is possible that $m_{i,0}$ differed from the other trials at that point (e.g. the system may not have fully equilibrated before the next clamp experiment was run).

%\marisa{also want to mention that this same calculation also shows that the initial conditions for the gating variables are identifiable in the scaled model from Theorem 2.1 (as these will be our $m_{i,0}/m_{i,\infty}$)}

\subsection{Applications}

%\subsubsection{Rescaling existing Hodgkin-Huxley models} 
\subsubsection{Rescaling the Sim-Forger model} 
To demonstrate the unidentifiability and identifiable combinations determined in Theorem \ref{th:ICcombos} using an HH-model applied in practice, we consider the Sim-Forger model of a suprachiasmatic nucleus (SCN) neuron \cite{sim2007modeling}. Extensions of this model have used the HH model to gain insight into the underlying mechanisms of timekeeping in the SCN \cite{Belle2009, Diekman2013}. The non-uniqueness of the steady-state parameters when the initial conditions are not known allows us to generate the same output from two different Hodgkin-Huxley style models using two different sets of initial conditions. The sodium current equation in this model is given by: $I_{Na} = g_{Na}(V-E_{Na})m^3 h$. The steady-state parameters $m_{\infty}$ and $h_{\infty}$ are given by the equations
\[m_{\infty} = \frac{1}{1 + \exp(-(V + 35.2)/7.9)}\]
\[h_{\infty} = \frac{1}{1 + \exp((V + 62)/5.5)}\]
with initial conditions $m_0 = 0.34$ and $h_0 = 0. 045$. If we rescale so that $\hat{m}_{\infty} = a m_{\infty}$, $\hat{h}_{\infty} = \frac{ h_{\infty}}{a^3}$, $\hat{m}_0 = 0.34 a$ and $\hat{h}_0 =\frac{ 0.045}{a^3}$, the model will produce identical output for $a \neq 0$. This is shown in Figure \ref{fig:unid}B, where the solid line shows the output for $a=1$ while the open circles shows the same output for $a = 2$. The two traces are identical.

%\subsection{Modeling Considerations}
%\mcolor{Next, we examine several examples and additional cases that may be relevant when applying these results in practice.}

\subsubsection{Fitting the time constants through a least squares approach}
Finally, while the proof of time constant identifiability in Theorem \ref{th:tau} does not immediately appear useful for parameter estimation, we next illustrate how the input-output equations obtained in the proof of Theorem \ref{th:tau} can be used to estimate the identifiable $\tau_i$'s. We demonstrate this using the two-gate HH model, with each gate appearing once: $I = g(V-E)mh$. We generated simulated voltage clamp data, $I(t)$ (with a timestep of 0.1 milliseconds), and the first and second derivatives $(I', I'')$ were estimated numerically from the data in MATLAB (from the slope of the line between first two data points). 

The resulting current trace and its derivatives were concatenated into a $T$-by-$9$ matrix, where $T$ is the number of time points composing $I$ and 9 is the number of distinct monomials in 
$(I-1)$, $I'$, and $I''$, of maximum degree two: 
{\small
\[A = 
\begin{pmatrix}
\vdots & \vdots & \vdots & \vdots & \vdots & \vdots & \vdots & \vdots & \vdots \\
I -1 & I' & I'' & (I -1)^2 & {I'}^2 & {I'}'^2 & (I -1)I' & I'I'' & (I-1) I'' \\ 
\vdots & \vdots & \vdots & \vdots & \vdots & \vdots & \vdots & \vdots & \vdots \\
\end{pmatrix}\]
}

These monomials should form an input-output equation with the appropriate coefficients; hence, the vector $x$ making up the null space of $A$ will give us these coefficients.

To ensure that $A$ \emph{has} a null space, we took its singular value decomposition and made the least singular value equal to zero. This corrected for any errors in the derivative computation and enforced the rank deficiency requirement. We then solved for $x$ from $Ax = 0$. From the theorem proving the identifiability of the time constants, the first three entries of $x$ should be the coefficients of the degree 2 polynomial with roots $-1/\tau_m$ and $-1/\tau_h$. We then found the roots of this 
polynomial using the aptly-named MATLAB function `roots'. 

Indeed, the roots did agree with the prescribed time constants. For preset $\tau_h = 4$ and $\tau_m = 22$, the time constants recovered in this way were $\tau_h =  3.9928$ and $\tau_m =  22.0072$. While noise will likely confound this process in real data, it nonetheless provides an interesting and novel way of fitting HH-style equations. As the other identifiability results predict, the steady-state parameters do not need to be known to estimate the time-constants. In addition, no initial guess of where the time constants lie in parameter space is needed to arrive at this estimate. Thus, even given the issues that may come with estimating the derivatives of $I(t)$ in the presence of noise, this approach may also be a useful way to obtain initial estimates of the $\tau_i$'s which are `in the ballpark', and then more conventional optimization approaches can be used.

\section{Conclusions}

In this work, we have shown that the time constants for the gating variables of a generalized HH-type model are identifiable, while the steady-state parameters are not---unless initial conditions are known. In the case where the initial conditions for the gating variables are not known, we have also shown how the steady state parameters form identifiable combinations, both as a single product with the conductance $g$, and as ratios with their initial conditions. We have further demonstrated that for ionic currents with two types of voltage-dependent gates, the number of each is an identifiable parameter.

Given the common use of parameter estimation to connect HH models with data in the literature \cite{Buhry2011,Lee2006,Csercsik2010,Willms1999,Fink2009,Willms2002,Hafner1981,Vavoulis2012}, these results may be directly useful for 
%Given the common use of parameter estimation to connect HH models with data in the literature \cite{many things!}, these results may be directly useful for 
%these are relevant facts to keep in mind while fitting and using electrophysiological models, p
fitting and using electrophysiological models. In particular, the points about initial conditions may be useful in practice, as matching both the initial condition and the steady-state parameter at once is an underdetermined problem. Previous work has noted the range of difficulties with estimation for HH-type models \cite{Buhry2011,Lee2006, Willms1999, Csercsik2010}, and this analysis may help to both explain and improve on some of these issues by explicitly laying out the identifiability properties of general HH-type models. 

In \cite{Csercsik2010}, the authors used differential polynomial reduction to show that the steady-state parameters are not identifiable. Rather than attempting that computation in the generalized case, we simply rescaled the equations to conclude that the parameters are not identifiable from the altered identifiability structure (one parameter fewer) of the equations. Rescaling in this way is a quick and easy way to begin to approach identifiability questions in the wild, and we hope it proves useful to those who are less comfortable with differential algebra. 

To show that the time constants were identifiable in the fully general case, we used a linear system that emerges from the structure of the HH equations and their derivatives. By solving the linear system and equating one entry in the resulting vector with the products of others, we generated an input-output equation. The coefficients of this input-output equation were the coefficients of the polynomial with roots at the time constants and sums of time constants. We were further able to demonstrate the identifiability of the time constants by using the method described in our proof to compute two time constants from simulated data.

In this way, our proof for the identifiability of the time constants also led to the development of a novel approach to estimation of the gating variable time constants, which does not require knowledge or estimation of the steady state parameters. The matrix form of the input-output equations allowed us to estimate parameters by considering the nullspace, which we illustrated using simulated data from the two-gate HH model case. By contrast, more standard estimation approaches would need to estimate the (potentially unidentifiable) steady state parameters and gating variable initial conditions in order to estimate the time constant parameters---this can be somewhat ameliorated by re-scaling the model by the steady state constants (e.g. as in the proof of Theorem \ref{th:unid}), or by knowledge of the initial conditions, but still requires additional information. This approach enabled us to arrive at a good approximation of the original parameters without needing to know or guess any other unknown parameters, including initial conditions. While we recognize that the presence of noise would confound this process, the idea behind it could prove useful in later work, perhaps in suggesting a starting point in parameter space for error-minimizing parameter searches.

This analysis has focused only on the identifiability of the Hodgkin-Huxley model from data obtained through voltage clamp; a natural extension for future work is to consider data taken from current clamp experiments, in which a current is applied and changes in voltage are recorded, and action potential clamp experiments, which are similar to voltage clamp except that instead of a \emph{constant} voltage being maintained via a feedback loop, the voltage is instead fixed to match an action potential.
Finally, the extensions of Hodgkin-Huxley are wide and varied and encompass much more than voltage-dependent gates acting independently. There is a broad literature of ion channel models out there that could likely benefit from inspection similar to this.

\section*{Acknowledgments} 
We would like to thank Danny Forger for his discussions with us about this work. This material is based upon work supported by the National Science Foundation Graduate Student Research Fellowship under Grant No. DGE 1256260.

\bibliography{idbib}
\bibliographystyle{ieeetr}

%\newpage

\end{document}